\documentclass[letterpaper, 10 pt, conference]{ieeeconf}  

\IEEEoverridecommandlockouts    

\usepackage{times}

\usepackage{amsmath}
\usepackage{amssymb}
\usepackage{amsthm}
\usepackage{mathtools}
\usepackage{bbm}
\usepackage{dsfont}

\usepackage{graphicx}
\usepackage{xcolor}

\usepackage{multirow}

\usepackage[algoruled, noline, noend, linesnumbered]{algorithm2e}

\usepackage{caption}

\usepackage{subcaption}

\usepackage{hyperref}

\usepackage{tikz}
\usepackage{cite}
\usepackage{booktabs}   
\usepackage{xspace}

\newtheorem{theorem}{Theorem}
\newtheorem{problem}{Problem}

\newtheorem{assumption}{Assumption}
\newtheorem{proposition}{Proposition}

\newtheorem{definition}{Definition}
\newtheorem{remark}{Remark}

\newcommand{\px}{\mathbf{x}}
\newcommand{\z}{\mathbf{z}}
\newcommand{\pu}{\mathbf{u}}
\newcommand{\pv}{\mathbf{v}}

\newcommand{\prop}{\mathrm{p}}
\newcommand{\str}{\pi}
\newcommand{\A}{\mathcal{A}}

\newcommand{\U}{\mathcal{U}}

\newcommand{\X}{\mathcal{X}}

\newcommand{\I}{\mathcal{I}}

\newcommand{\reals}{\mathbb{R}}
\newcommand{\nats}{\mathbb{N}}

\DeclareMathOperator*{\argmin}{\arg \min}

\title{Learning-Based Shielding for Safe Autonomy under Unknown Dynamics
}
\author{Robert Reed and Morteza Lahijanian%
\thanks{This work was supported by Air Force Research Lab (AFRL) under agreement number FA9453-22-2-0050.}%
\thanks{Authors are with the Department of Aerospace Engineering Sciences at University of Colorado Boulder, Boulder, Colorado, USA. {\tt\small \{Robert.Reed-1, Morteza.Lahijanian\}@colorado.edu }}
}
\date{August 2024}

\begin{document}

\maketitle
\begin{abstract}
    \emph{Shielding} is a common method used to guarantee the safety of a system under a black-box controller, such as a neural network controller from deep reinforcement learning (DRL), with simpler, verified controllers. Existing shielding methods rely on formal verification through Markov Decision Processes (MDPs), assuming either known or finite-state models, which limits their applicability to DRL settings with unknown, continuous-state systems. This paper addresses these limitations by proposing a data-driven shielding methodology that guarantees safety for unknown systems under black-box controllers. The approach leverages Deep Kernel Learning to model the systems' one-step evolution with uncertainty quantification and constructs a finite-state abstraction as an Interval MDP (IMDP). By focusing on safety properties expressed in safe linear temporal logic (safe LTL), we develop an algorithm that computes the maximally permissive set of safe policies on the IMDP, ensuring avoidance of unsafe states. The algorithms soundness and computational complexity are demonstrated through theoretical proofs and experiments on nonlinear systems, including a high-dimensional autonomous spacecraft scenario. 
\end{abstract}

\section{Introduction}
    \label{intro}
    A \emph{shield} or \emph{Simplex architecture} is a control framework used in \emph{safety-critical} systems where a high-performance primary controller is ``shielded" or backed up by a simpler, verified controller to ensure safety \cite{bloem2015shield, alshiekh2018safe}. One popular application of this architecture is in neural network (NN) controllers derived from deep reinforcement learning (DRL), where the goal is to balance the performance of DRL with safety by intervening when DRL actions may lead to unsafe outcomes \cite{harris2020spacecraft, nazmy2022shielded, Reed2024Shield}. Shielding is often implemented through formal verification, modeling the evolution of safety-relevant states as a \emph{Markov Decision Process} (MDP), with model-checking techniques used to identify safe policies. In many works, the safety MDP is either assumed to be known or simplified to a finite set of states \cite{alshiekh2018safe, bloem2015shield, yang2023safe, Reed2024Shield, jansen2020safe}. However, these assumptions are restrictive in DRL settings: (i) many real-world scenarios involve continuous state spaces where finite-state MDPs cannot accurately capture system dynamics, and (ii) DRL is typically applied to unknown systems, making it unrealistic to assume prior knowledge of the safety MDP. Hence there is a lack of literature on guaranteeing the safety of unknown, continuous state systems under DRL policies.
This paper aims to bridge the gap by developing a data-driven shielding methodology. 

In this paper, we propose an algorithmic approach to generate shields from data, ensuring the safety of unknown systems operating under black-box controllers. The key insight of our method is that the data used to train DRL controllers can also be leveraged to construct a finite abstraction of the system with safety guarantees. Specifically, we employ \emph{deep kernel learning} \cite{reed2023promises, Ober2021}
to model the one-step evolution of the system with rigorously-quantified uncertainty in the learning process. Using this model, we then generate a finite-state abstraction of the system as an interval MDP (IMDP), a formalism that enables the capture of state space discretization error as well as learning uncertainty.

We focus on safety properties expressed in \emph{safe linear temporal logic} (safe LTL) \cite{kupferman2001model}, an expressive formal language whose formulas are evaluated over infinite trajectories. An important property of safe LTL specifications is that finite trajectories are sufficient to violate them. Based on this, we develop an algorithm that computes the maximal set of safe policies on the IMDP abstraction. This is achieved by converting the safety problem into a reachability problem by composing the IMDP with the complement of the safe LTL specification. Unsafe actions are then systematically removed from the composed model, resulting in a maximally permissive set of controllers that guarantee avoidance of unsafe states.
We prove that our algorithm is both sound and empirically demonstrates computational efficiency. Experimental validation on various nonlinear systems under different safety specifications, including a 6-dimensional realistic spacecraft, demonstrates the scalability of the method and the ability of the generated shields to support complex behaviors, confirming the theoretical claims.

The \emph{main} contribution of this paper is the data-driven shield construction for unknown, continuous state systems under safe LTL specifications with probabilistic guarantees. Other contributions include: (i) an algorithm to find the maximally permissive set of safe policies on an IMDP model, (ii) proofs that the set of policies returned by the algorithm guarantee the safety of the \emph{unknown, continuous state} system, and (iii) demonstrations of the efficacy of our approach in several case studies, including a realistic, complex spacecraft system using a flight-approved physics simulator.

\subsection{Related Works}

Prior works for shielding can be generally grouped based on the model assumptions used to generate safety guarantees: (a) a known model of the safety states, (b) an unknown model that is represented with a learned finite state MDP.

Work \cite{alshiekh2018safe} considers shielding for a known finite state MDP with deterministic guarantees.
Despite using an MDP, a probabilistic model, the shield is generated through a game-based approach that removes transition probabilities resulting in a shield where actions are only allowed if there is zero probability under all paths of reaching the unsafe states. 
This built off of Work \cite{bloem2015shield}, which considers a system modelled as a known finite state reactive systems, i.e. no probability is included in the transitions. Both of these works consider safety specifications which are defined by a set of allowable traces, which is a generalization of the safe LTL specifications used in our approach. However both of these methods are too restrictive as they do not allow for probabilistic violations of safety and may result in shields where no actions are considered safe. 

Work \cite{yang2023safe} developed a method of shielding with probabilistic constraints for continuous state systems, with the probabilities being based on a known continuous state MDP model of the safety states. While this model enables the guarantees to hold for the full continuous state system, the model choice makes it intractable to consider guarantees beyond a single time step. This also restricts the guarantees to one step reachability and prohibits the use of more complex safety specifications that can be defined with LTL.

When the system model is unknown, prior works have generated empirical MDP representations from data.  Work \cite{jansen2020safe} considers shielding to guarantee that a continuous state system can avoid collision 
with adversarial agents with unknowns dynamics. However, Work \cite{jansen2020safe}
generates a finite state MDP abstraction from observations of how the adversaries act, which limits the guarantees to the MDP abstraction rather than the true system. 
Work \cite{jansen2020safe} also limits safety properties to finite horizons, specifically restricting the method to provide only one step guarantees on safety.
Work \cite{Reed2024Shield} similarly generates an empirical MDP abstraction of a continuous state system. Work \cite{Reed2024Shield} defines three different shielding methods to provide guarantees for arbitrary safe LTL specifications. However, the formulation for both Work \cite{jansen2020safe} and \cite{Reed2024Shield} results in safety guarantees that only hold for the empirical MDP and not necessarily the true system.

An alternate approach to identifying safe policies is through the use of Control Barrier Functions \cite{mazouz2024data, wajid2022formal, jagtap2020formal, santoyo2021barrier}. While these approaches are shown to work with unknown systems, many of them only identify a single controller that can guarantee that the system remains in an invariant set for a finite time horizon. As typical in shield construction, work \cite{mazouz2024data} identifies the maximally permissive set of safe controllers. Yet, while \cite{mazouz2024data} produces similar results to this work, they restrict the notion of safety to identifying an invariant set for the system and require an initial set from which the system is deployed.

To the best of our knowledge, no prior work generates a shield that can provide guarantees on the safety of an unknown continuous state system for arbitrary safety specifications. Earlier works either provide guarantees for a known system - with only finite state systems allowing an infinite horizon safety specification - or provide guarantees on an empirical finite state MDP abstraction of an unknown system. Hence, this is the first work that enables infinite horizon guarantees on safety for unknown, continuous state systems.

\section{Problem Formulation}
    \label{problem}
    
\subsection{System Model}

Consider a discrete-time stochastic system:
\begin{align}
    \px(k+1) = f(\px(k), \pu(k)) + \pv(k), \label{true_dynamics}
\end{align}
where $\px(k) \in \reals^n$ is the state, $\pu(k) \in U$ is the control, $U = \{a_1, \ldots, a_{|U|}\}$ is a finite set of actions or control laws, $\pv(k) \in 
V = \{v \in \reals^n \mid \|v\|_{\infty} \leq \sigma_v \in \mathbb{R}_{> 0} \}$
is an additive 
random variable
with bounded support and probability density function $p(\pv)$, 
and $f: \reals^n \times U \rightarrow \reals^n$ is an \emph{unknown} function (vector field). Intuitively, System \eqref{true_dynamics} represents a switched system with bounded process noise and stochastic dynamics, a typical model used for reinforcement learning. We assume a dataset $D = \{({x}_i, {u}_i, {x}_i^+)\}_{i=1}^{m}$ of i.i.d. samples of System~\eqref{true_dynamics} can be gathered, where ${x}_i^+$ is a realization of the one time-step evolution of System~\eqref{true_dynamics} from $x_i$ under action $u_i$.

We impose a standard smoothness (well-behaved) assumption \cite{Chowdhury, abbasi2013online} on $f$ since it is unknown and must be learned from data. 
\begin{assumption} [RKHS Continuity] \label{assumption:1}
For a compact set $C \subset \reals^n$, let $\kappa: \reals^n \times \reals^n \rightarrow \reals_{\geq 0}$ be a given continuously differentiable kernel and $\mathcal{H}_\kappa(C)$ the reproducing kernel Hilbert space (RKHS) corresponding to $\kappa$ over $C$ with induced norm $\|\cdot\|_\kappa$. Then for each $a \in U$ and $i \in [n]$, $f^{(i)}(\cdot, a) \in \mathcal{H}_\kappa(C)$, where $f^{(i)}$ denotes the $i$-th component of $f$. Moreover, there exists a constant $B_i \geq 0$ s.t. $\|f^{(i)}(\cdot, a)\|_\kappa \leq B_i$.
\end{assumption}
Assumption \ref{assumption:1} implies that each $f^{(i)}(x, a)$ can be represented as a weighted sum of $\kappa(x, \cdot)$. Intuitively, we can choose $\kappa$ such that $\mathcal{H}_\kappa$ is dense in the space of continuous functions, e.g., with the squared exponential kernel function \cite{steinwart2001influence}. This assumption allows us to use Gaussian Process (GP) regression to estimate $f$ and bound the error on the estimates, even with non-Gaussian, bounded noise.

In reinforcement learning, System \eqref{true_dynamics} is described as a continuous-state Markov Decision Process (MDP).
\begin{definition} [MDP] \label{def:mdp}
    A \emph{Markov Decision Process} (MDP) is a tuple $M = (X, X_0, U, T, AP, L)$, where $X \subset \reals^n$ is a compact set of states, $U$ is a set of actions, $X_0 \subseteq X$ is a set of initial states, $T: X \times U \times \mathcal{B}(X) \rightarrow [0, 1]$ is a transition kernel, where $\mathcal{B}$ is a Borel set\footnote{The Borel set defines an open set of states in $X$, hence transition probabilities can be assigned.}, $AP$ is a set of atomic propositions that are related to the agents task or safety, and $L: X \rightarrow 2^{AP}$ is a labeling function that assigns a state $x \in X$ to a subset of $AP$.
\end{definition}
Let $W \subset \reals^n$ be a Borel measurable set. Then, the transition kernel $T$, which defines the probability that $\px(k+1) \in W$ given $\px(k) = x \in X$ and action $a\in U$, is
\begin{align}
    T(x, a, W) = \int_W (f(x, a) +  \pv) p(\pv) d\pv.
\end{align}

An infinite trajectory of this system is written as $\omega_\px = x_0 \xrightarrow{u_0} x_1 \xrightarrow{u_1} ...$ where each $u_i \in U$ and $x_i \in X$ and the set of all finite and infinite trajectories are denoted as $\Omega_\px^{fin}$ and $\Omega_\px$. A \emph{policy} $\str: \Omega_\px^{fin} \rightarrow U$ maps a finite trajectory $\omega_\px \in \Omega_\px^{fin}$ onto an action in $U$. 
A policy is called \textit{stationary} or \textit{memoryless} if it depends only on the last state in the finite path.
Under a policy $\str$, the paths of $M$ have well defined probability measures \cite{lahijanian2011control}, and the MDP $M$ under $\str$ is denoted as $M^\str$ with the sets of finite and infinite trajectories denoted as $\Omega_\px^{fin,\str}$ and $\Omega_\px^\str$.

Typically, we are interested in the temporal properties of $\px$ in the compact set $X \subset \reals^n$ with respect to a set of regions $R = \{r_1, \ldots, r_l\}$, where $r_i \subseteq X$, e.g., we would like the system to reach several locations of interest in a particular order and avoid some keep-out zones. The set of atomic proposition $AP = \{\prop_1, \ldots, \prop_l\}$ is then defined according to these regions, where $\prop_i$ is true iff $\px \in r_i$. 
Using the labeling function $L$,
then, the \textit{observation trace} of trajectory $\omega_\px$ is $\rho = \rho_0 \rho_1 \ldots$, where $\rho_i = L(\omega_\px \! (i))$ for all $i \in \nats$, is the sequence of observed atomic proposition sets.

To formally specify the safety specification, we use \emph{syntactically safe} LTL \cite{kupferman2001model}, which is a language that can express safety requirements that a system must not violate with a set of Boolean connectives and temporal operators.
\begin{definition}[Safe LTL]
    \label{def:sLTL}
    Given a set of atomic propositions $AP$, a safe LTL formula is inductively defined as
    \begin{equation*}
        \varphi = \prop \mid \neg \prop \mid \varphi \land \varphi \mid \varphi \lor \varphi \mid \bigcirc \varphi \mid \square \varphi 
    \end{equation*}
    where $\prop \in AP$,  $\neg$ (``not''), $\land$ (``and''), $\lor$ (``or'') are Boolean connectives, and $\bigcirc$ (``next'') and $\square$ (``globally'') are temporal operators.
\end{definition}
From $\bigcirc$ and Boolean connectives, we can construct Bounded Until ($\U^{\leq k}$) and Bounded Globally ($\square^{\leq k}$) on $\prop$ and $\neg \prop$ \cite{baier2008principles}. As an example with $AP = \{\prop_1, \prop_2\}$, a specification ``Every time $\prop_1$ is visited, it cannot be visited again until either $\prop_2$ has been seen or 10 time steps have passed, whichever occurs first" can be written in safe LTL as $\varphi = \square(\prop_1 \rightarrow (\neg \prop_1 \U^{\leq 10} \prop_2) \lor (\square^{\leq 10} \neg \prop_1))$.
The semantics of safe LTL are defined over infinite traces \cite{kupferman2001model}. An infinite trajectory $\omega_\px \in \Omega_\px$ satisfies a safe LTL formula $\varphi$, written as $\omega_\px \models \varphi$, if its observation trace satisfies $\varphi$. For MDP $M$, trajectories have an associated probability measure hence the satisfaction of $\varphi$ is probabilistic. The probability of $M$ under policy $\str$ satisfying $\varphi$ is then defined as
\begin{align}
    \mathbb{P}(M^{\str} \models \varphi) = \mathbb{P}(\omega_\px \in \Omega_\px^\str \mid \omega_\px \models \varphi).
\end{align}

Given unknown System~\eqref{true_dynamics} and safety formula $\varphi$, our goal is to generate a maximal set of policies, called a shield, that guarantee no violation to $\varphi$.

\subsection{Problem Statement}
Our problem can then be stated as follows.
\begin{problem} [Continuous State Shielding] \label{prblm:1} 
    Given a dataset $D = \{({x}_i, {u}_i, {x}_i^+)\}_{i=1}^{m}$ of i.i.d. samples of a continuous state System~\eqref{true_dynamics} 
    and safe LTL formula $\varphi$,
    identify a maximal set of policies $\Pi_S$
    such that the probability of the system under policy $\str \in \Pi_S$ violating $\varphi$ is less than $p$, i.e., 
    \begin{align*}
        \qquad \qquad \qquad & \mathbb{P}(M^{\str} \models \neg \varphi) < p && \forall \str \in \Pi_S.
    \end{align*}
\end{problem}

Prior works generate shields under a two major assumption: (i) a \emph{finite state} MDP that describes the safety aspects of the system exists, and (ii) this MDP is known a priori. However, neither of these assumptions hold in our setting. In particular, a finite state MDP \emph{cannot} accurately model a continuous state system. Hence, traditional shielding techniques cannot be used to provide guarantees on the continuous state system. Similarly, even if traditional shielding techniques could be used, the system model is unknown, hence generating a \emph{correct} finite state model is a challenge.

Our approach to Problem~\ref{prblm:1} is based on creating a finite abstraction of System \eqref{true_dynamics} as an Interval MDP (IMDP) from $D$ using GP regression, specifically Deep Kernel Learning (DKL) \cite{reed2023promises}. DKL is an extension of GP regression that augments the input to a base kernel with a neural network prior. This formulation is shown to enable more accurate mean predictions, smaller posterior variance predictions, and faster abstraction generation \cite{reed2023promises}. We then reason about all possible traces of the system over the IMDP and reject actions that violate our desired probability threshold. Our result is the maximally permissive set of policies $\Pi_S$ that contain no strategies that violate $\varphi$. 

\begin{remark}
    We emphasize that in our formulation the process noise has a bounded support, as when noise has unbounded support it is impossible to guarantee safety (LTL formula $\varphi$) for an infinite time horizon since, for a bounded set $X$, the probability of remaining in $X$ under infinite support noise is zero over infinite trajectories. When used in a RL setting, $\Pi_S$ constitutes a shield as in Figure~\ref{fig:shield_arch}, where an agents action is only corrected if the choice is not in the set of shield policies. 
\end{remark}

\begin{figure}[t]
    \centering
    \includegraphics[width=0.75\linewidth]{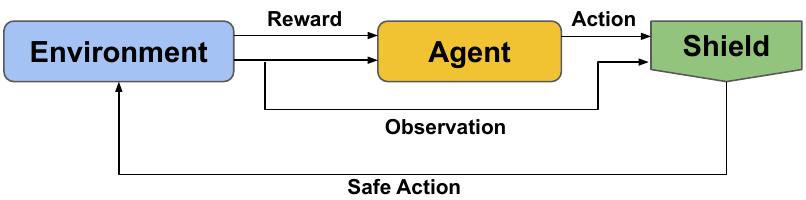}
    \caption{Post-Posed Shielded RL architecture. In our scenario, the shield can be taken as $\Pi_S$.}
    \label{fig:shield_arch}
\end{figure}

\section{Methodology}
    \label{method}

For our approach, we first use DKL to model dynamics $f$ from dataset $D$ and then construct a finite state abstraction in the form of an Interval MDP (IMDP) where the transition probabilities account for the learning error of DKL, the noise distribution, and the discretization of the continuous state space. 
We utilize methods that make use of GP regression when process noise can have bounded support, e.g., methods defined in \cite{jackson2021formal}. These methods enable us to make use of DKL to generate highly accurate predictions with tight regression error bounds. 

\subsection{Deep Kernel Learning}
DKL is an extension of GP regression. A GP is a collection of random variables, such that any finite collection of those random variables is jointly Gaussian \cite{rasmussen:book:2006}. GP regression places a prior over $f$ through use of a kernel function $\kappa: X \times X \rightarrow \reals_{\geq 0}$. Then, conditioned on a dataset $D$, a GP produces dimension-wise posterior predictions at point $x^*$ with a mean $\mu^{(i)}_{D}(x^*)$ and variance $\sigma^{(i), 2}_{D}(x^*)$ which are derived as: 
\begin{align}
    &\mu^{(i)}_{D}(x^*) = K_{x^*, \X}(K_{\X,\X} + \sigma_n^2 I)^{-1} Y, \label{mean_pred}\\
    &\sigma^{(i), 2}_{D}(x^*) = K_{x^*, x^*} - K_{x^*, \X}(K_{\X,\X} + \sigma_n^2 I)^{-1} K_{\X, x^*}, \label{covar_pred}
\end{align}
where $\X = [x_1, \ldots, x_{m}]^T$ and $Y = [x_1^{(i),+}, \ldots, x_m^{(i),+}]^T$ are respectively the vectors of input and output data, $K_{\X,\X}$ is a matrix whose $l$-th row and $j$-th column is $\kappa(x_l, x_j)$, and $K_{x,\X} = K_{\X, x}^T$ are vectors whose $l$-th entry is $\kappa(x, x_l)$. 
In DKL, a base kernel, e.g., the squared exponential kernel $\kappa(x,x') = \sigma_s \exp{-\|x-x'\|/2l^2}$, is composed with a neural network as $\kappa_{DKL}(x, x') =\kappa(\psi_w(x), \psi_w(x'))$, where $\psi_w: \reals^n \rightarrow \reals^s$ is a neural network parameterized by weights $w$. When $\psi_w$ is trained to model $f$, such a formulation can significantly improve the representational power of GPs while retaining an analytical posterior prediction.

DKL enables us to predict the one time-step propagation of System~\eqref{true_dynamics} with rigorously quantified error. This is, however, not sufficient for computation of a policy over trajectories of arbitrary length. To that end, we construct a finite abstraction of System~\eqref{true_dynamics} using the one time-step predictions and reason over the abstraction.

\subsection{IMDP Abstraction}
Our abstraction of System~\eqref{true_dynamics} is an IMDP \cite{givan2000bounded, Lahijanian:TAC:2015}.
\begin{definition} [IMDP] \label{def:imdp}
    An Interval MDP (IMDP) is a tuple $\I = (Q, U,\hat{P}, \check{P}, AP, L_\I)$, where $Q$ is a finite set of states, $U, AP$ are as in Def.~\ref{def:mdp}, $\check{P},\hat{P}: Q \times U \times Q \rightarrow [0, 1]$ are functions that define the lower and upper bounds, respectively, of the transition probability from state $q\in Q$ to state $q'\in Q$ under action $a \in U$, and $L_{\I }: Q \to 2^{AP}$ is a labeling function that assigns to each state $q \in Q$ a subset of $AP$.
\end{definition}
An adversary $\Delta: Q \times U \times Q \rightarrow [0,1]$ is a function which picks a valid transition probability distribution from $\check{P}, \hat{P}$ such that 
$\sum_{q' \in Q} \Delta(q, a, q') = 1$ and 
$\check{P}(q, a, q') \leq \Delta(q, a, q') \leq \hat{P}(q, a, q')$. An IMDP $\I$ under adversary $\Delta$, denoted as $\I^{\Delta}$ becomes an MDP. Then under an adversary $\Delta$ and policy $\str$, denoted $\I^{\str, \Delta}$, the probability of satisfying $\varphi$ is well defined.
Then, any policy that can guarantee safety under worst-case adversary on the IMDP provides safety guarantee for the continuous state system.

\textbf{States and Labels}. 
We first partition the space $X$ into a set of regions $\overline{Q} = \{q_0, q_1, \ldots, q_{|\overline{Q}|}\}$ such that the partitions respect the regions of interest $R$. Then, for each region of interest $r \in R$ there exists $Q_r \subseteq \overline{Q}$ such that $r = \cup_{q\in Q_r}q$ and for all $i \neq j$ $q_i \cap q_j = \emptyset$. We then define an extra region $q_u$ where $q_u = \reals^n / X$ and the states $Q$ of the IMDP model are defined as $Q = \overline{Q} \cup \{q_u\}$. With an abuse of notation we use $q$ to denote both the IMDP state and the corresponding region of space, i.e., $q \in Q$ and $q \subset \reals^n$. This format enables $L(x) = L(x')$ for all $x, x' \in q$, hence the IMDP labelling function $L_{\I }(q) = L(x)$.
Finally, the set of IMDP actions is given by $U$ and we assume that for each $q \in Q$ all actions are available.

\textbf{Transition Probabilities}. We start by learning a model of $f$ with DKL using dataset $D$, allowing us to estimate each component of $f$ with a $\mu_D^{(i)}$ as in \eqref{mean_pred}.
DKL is used to propagate region $q$ under action $a$, with the consideration of all possible process noise in a set $c \subseteq V$ as
\begin{multline}
    \text{Post}(q, a, c) = \{w \mid w^{(i)} = \mu_{D}^{(i)}(x) + v^{(i)}, \\
    x \in q, v \in c, i \in [1, \ldots, n] \}. \label{post_eq}
\end{multline}
As Post is generated from a learning model, we must also reason about the predictive error of DKL. Assumption \ref{assumption:1} allows for the generation of uniform error bounds on $\mu_D^{(i)}$ \cite{abbasi2013online, Chowdhury, reed2024error}, particularly,
\begin{proposition}
    [\cite{reed2024error}, Theorem 1] Let $B_i \geq \|f^{(i)}\|$, $G = (K_{\X,\X} + \sigma_n^2I)^{-1}$, $W_x = K_{x,\X}G$, $\gamma \leq f^{(i)}(\X)^T G f^{(i)}(\X)$, and
    $\lambda_x = 4\sigma_v^2K_{x,\X}G^2 K_{\X,x}.$ Then if $\|\pv\|_{\infty} \leq \sigma_v$,
    \begin{align}
        \mathbb{P}\Big(|\mu_{D}^{(i)}(x) - f^{(i)}(x, a)| \leq \epsilon^{(i)}(\delta, x) \Big) \geq 1 - \delta,
    \end{align}
    where $\epsilon^{(i)}(x,\delta) = \sigma_{D}^{(i)}(x)\sqrt{B_i^2 - \gamma} + \sqrt{\frac{\lambda_x}{2}\text{ln}\frac{2}{\delta}}$
\end{proposition}
For a region $q$, let $\epsilon^{(i)}_q(\delta) = \sup_{x \in q}\epsilon^{(i)}(\delta, x)$, and let the worst case dimension-wise regression error over a region $q$ and action $a$ be $e^{(i)}(q,a) = \sup_{x \in q} |\mu_{D}^{(i)}(x) - f^{(i)}(x, a)|$. Then,
\begin{align}
    \mathbb{P}\Big(e^{(i)}(q,a) \leq \epsilon^{(i)}_q(\delta) \Big) \geq 1 - \delta.
\end{align}
Then the transition kernel from $x \in q$ under action $a$ to $q'$ can be lower and upper bounded according to the following Proposition.
\begin{proposition}
[\cite{skovbekk2023formal, jackson2021formal}, Theorems 1] Let $\underline{q}(\epsilon)$ and $\bar{q}(\epsilon)$ be regions that are obtained by shrinking and expanding each dimension of $q$ by the corresponding scalars in $\epsilon$, and let $C$ be a partition of $V$ such that $\sum_{c \in C} \int_c p(\pv)d\pv = 1$. Then,
\begin{multline}
    \min_{x\in q}T(x,a, q') \geq \sum_{c\in C} \bigg( \left(1 - \mathds{1}_{X \setminus \underline{q'}(\epsilon_q)}(\text{Post}(q, a, c))\right) \\
    \prod_{i=1}^{n} \mathbb{P}\left(e^{(i)}(q,a) \leq \epsilon_{q}^{(i)}(\delta)\right) \bigg) \mathbb{P}(c)
    \label{eq:lower-bound}
\end{multline}
\begin{multline}
    \max_{x\in q}T(x,a, q') \leq \sum_{c\in C} \bigg( 1 - \prod_{i=1}^{n} \mathbb{P}\left(e^{(i)}(q,a) \leq \epsilon_{q}^{(i)}(\delta)\right)  \\
    \left(1 - \mathds{1}_{\overline{q'}(\epsilon_q)}(\text{Post}(q, a, c)\right) \bigg)  \mathbb{P}(c)
    \label{eq:upper-bound}
\end{multline}
where $\mathds{1}_W(H) = 1$ if $W \cap H \neq \emptyset$ and $0$ otherwise.
\end{proposition}

Finally the IMDP transition probability bounds are set to the bounds in \eqref{eq:lower-bound}-\eqref{eq:upper-bound}:
\begin{align}
    \check{P}(q, a, q') \leq \min_{x\in q}T(x,a, q'), \\
    \hat{P}(q, a, q') \geq \max_{x\in q}T(x,a, q'). 
\end{align}
The formulation of these intervals accounts for uncertainty due to modelling error, discretization of a continuous state space, and process noise which ensures that the IMDP abstraction contains the true MDP. Hence we can use the IMDP to generate a shield and naturally extend the guarantees to the unknown system.

\subsection{Shield Algorithm}

Note that the safe LTL formula $\varphi$ is a complement of syntactically safety (co-safe) LTL formula, which can be accepted in finite time \cite{kupferman2001model}.  Hence, from $\neg \varphi$, a deterministic finite automaton can be constructed that precisely accepts the set of traces that satisfy $\neg \varphi$ (violate $\varphi$)~\cite{kupferman2001model}.

\begin{definition} [DFA]  \label{def:dfa} 
A \emph{deterministic finite automaton} (DFA) constructed from a co-safe LTL formula $\neg \varphi$ is a tuple $\A = (Z, z_0, 2^{\Pi}, \delta, Z_f)$, where $Z$ is a finite set of states, $z_0 \in Z$ is an initial state, $2^{\Pi}$ is a finite set of input symbols, $\delta : Z \times 2^{\Pi} \rightarrow Z$ is a transition function, and $Z_f \subseteq Z$ is the set of final (accepting) states.
\end{definition}

A finite run on $\A$ is a sequence of states $\z = z_0z_1...z_{n+1}$ induced by a trace $\rho = \rho_0\rho_1...\rho_n$ where $\rho_i \in 2^{\Pi}$ and $z_{i+1} = \delta(z_i, \rho_i)$. A finite run is accepting if, for some $i\leq n$, $z_i \in Z_f$. If a run is accepting, its associated trace is accepted by $\A$. 
The set of all traces that are accepted by $\A$ is called the language of $\A$. The language of $\A$ is equal to the language of $\neg\varphi$, i.e., trace $\rho$ is accepted by $\A$ iff $\rho \models \neg\varphi$.
With a DFA constructed from $\neg \varphi$ and an IMDP $\I$, we can generate a product IMDP $\I_{\A}$, which captures 
the paths of $\I$ whose observation traces satisfy $\neg \varphi$ (violate safety specification $\varphi$).

\begin{definition}
    [Product IMDP] The product IMDP $\I_\A = \I \times \A$ is a tuple $\I_\A = (\bar{X}, U, \check{\Delta}, \hat{\Delta}, F)$ where $\bar{X} = Q \times Z$ is a set of product states, $U$ is as in Definition \ref{def:imdp}, $\check{\Delta}, \hat{\Delta}: \bar{X} \times U \times \bar{X} \rightarrow [0, 1]$ are such that $\check{\Delta}((q,z), a, (q',z')) = \check{P}(q,a,q')$ if $z' = \delta(z, L_\I(q))$ and 0 otherwise, similarly for $\hat{\Delta}$, and $F = Q \times Z_f$ is the set of final states. 
\end{definition}

We overload notation for an adversary on the product IMDP as $\Delta: \bar{X} \times U \times \bar{X} \rightarrow [0,1]$. 
With the Product IMDP, we can then find a set of safe policies that can guarantee no violations to $\varphi$ with at least probability $1-p$. Our goal is to find the maximal set of policies, which we refer to as a shield, that can guarantee safety specification $\varphi$.
We note that stationary policies on $\I_\A$ are sufficient to optimize (maximize) for the probability of satisfying $\neg \varphi$, 
i.e., there exists stationary $\str^*$ such that  $$\max_{\Delta} P(\I_\A^{\str^*,\Delta} \models \neg\varphi) = \max_{\str} \max_{\Delta} P(\I_\A^{\str,\Delta} \models \neg\varphi).$$  
Our approach is based on removal of such optimal policies until $\max_{\pi} \max_{\Delta} P(\I_\A^{\str,\Delta} \models \neg\varphi) \leq p$.
We avoid the computational intractability that can arise from searching for history dependent policies by removing the set of stationary policies on the product that can satisfy $\neg \varphi$. 

We make use of value iteration to find the set of actions that maximize the probability of satisfying $\neg \varphi$ and iteratively remove these actions from the product until we reach a fixed point. Algorithm \ref{alg:action_removal_2} shows how to find this set of safe policies. $\Delta$ is chosen such that $\sum_{\bar{x}' \in \bar{X}} \Delta(\bar{x},a,\bar{x}') V^{k}(\bar{x}')$ is maximized for each pair $(\bar{x},a) \in \bar{X} \times U$ and can be obtained through an ordering of states in $\I_\A$ according to $V^k$ and the transition probability intervals.

\begin{algorithm}[t]
\caption{Maximally Permissive Shield}\label{alg:action_removal_2}
\SetAlgoLined
\KwIn{$\bar{X}, F, U, p, \epsilon$}
\KwOut{$V^{k}, \Pi_S^k$}
$\Pi_S^0(\bar{x}) \gets U$\;
$V^0(\bar{x}) \gets 0 \; \forall \bar{x} \in \bar{X}$\;
$V^0(\bar{x}) \gets 1 \; \forall \bar{x} \in F$\;
$k \gets 0$\;
$\text{converged} \gets 0$\;
\While{not converged}{
    $\Pi_S^{k+1}(\bar{x}) \gets \{a \in \Pi_S^{k}(\bar{x})  \mid \sum_{\bar{x}' \in \bar{X}} \Delta(\bar{x},a,\bar{x}') V^{k}(\bar{x}') < p\}$\;
    \If{$\Pi_S^{k+1}(\bar{x}) = \emptyset$}{
        $\Pi_S^{k+1}(\bar{x}) \gets \argmin_{a \in \Pi_S^{k}(q)} \sum_{\bar{x}' \in \bar{X}} \Delta(\bar{x},a,\bar{x}') V^{k}(\bar{x}')$\;
    }
    \If{$\Pi_S^{k+1} \neq \Pi_S^{k}$}{
        $V^{k+1} \gets V^0$\;
        \tcp{Restart value iteration}
    }
    \Else{
        $V^{k+1}(\bar{x}) \gets \max_{a \in \Pi_S^{k+1}(\bar{x})} \sum_{\bar{x}' \in \bar{X}} \Delta(\bar{x},a,\bar{x}') V^{k}(\bar{x}')$\;
        \If {$|V^{k+1}(\bar{x}) - V^k(\bar{x})| < \epsilon \; \forall \bar{x} \in \bar{X}$}{
            $\text{converged} \gets 1$\;
        }
    }
    $k \gets k + 1$\;
}
\Return{$V^{k}, \Pi_S^k$}\;
\end{algorithm}

The algorithm takes in the product states $\bar{X}$, accepting states $F$ (i.e., the set of states that satisfy $\neg\varphi$), the action set $U$, a violation probability $p$, and a convergence threshold $\epsilon$ and returns the set of actions available at each state $\Pi_S$ and the worst case probability of being unsafe $V$ at each product state under any policy contained in the shield $\Pi_S$.

Lines 1-3 initialize the shield as allowing all actions at each state and the value as 1 in the accepting states of $\neg \varphi$ and 0 everywhere else. Line 13 updates the values at each state, choosing the value and $\Delta$ corresponding the action that maximizes the probability of reaching $F$, which results in a monotonically increasing value. Line 7 removes actions from the shield that exceed the allowed probability of satisfying $\neg \varphi$. Lines 8-9 ensure that each state retains its safest action in the event no action satisfies the threshold. Note that while this allows the shield to retain violating actions, the associated value allows us to identify which states are safe. This also ensures that the probability of the state satisfying $\neg \varphi$ can be calculated in a non-conservative manner under the remaining strategies.
Line 10 then checks if the shield has changed from the prior time step; if so, the value iteration is restarted with the new constrained action space. 

The following theorem shows that Algorithm~\ref{alg:action_removal_2} is guaranteed to terminate in finite time.

\begin{theorem}
    [Run Time Complexity] With state space $\bar{X}$ and action space $U$, Alg. \ref{alg:action_removal_2} is guaranteed to terminate in finite time and has a run time complexity of $\mathcal{O}(|\bar{X}|^2 \cdot (|U|!))$. \label{thm:run_time}
\end{theorem}

\begin{proof}
    It is well known that a traditional non-discounted Value Iteration algorithm has a computational complexity of $\mathcal{O}(|\bar{X}|^2|U|)$ and terminates in finite time \cite{sutton2018reinforcement}. Algorithm \ref{alg:action_removal_2} iteratively reduces the size of the action space $|U|$ and restarts value iteration under the new constrained space, which increases complexity to $|U|!$ as in the worst case the algorithm would remove only a single action in each iteration until only one action remains. This would result in value iteration with only one action at each state, which is guaranteed to terminate in finite time.
\end{proof}
While Algorithm \ref{alg:action_removal_2} has an exponential complexity in $|U|$, typically $|U| \ll |\bar{X}|$ and it is unlikely to reach this complexity as supported by empirical evidence in our experiments. Similarly while the worst case complexity results in removing all actions, Algorithm \ref{alg:action_removal_2} is guaranteed to return the maximal set of policies that can guarantee satisfaction of $\varphi$.

\begin{theorem} [Maximal Set of Safe Policies]
    Let $V^k$ and $\Pi_S^k$ be the output of Alg.~\ref{alg:action_removal_2}, and define $S = \{\bar{x} \mid V^k(\bar{x}) < p\}$.
    Then, from every $\bar{x} \in S$, there exists no policy in the shield $\Pi_s^k$ that allows for the satisfaction of $\neg \varphi$ with a probability $\geq p$ under any adversary $\Delta$, i.e.,
    $$\mathbb{P}(\I_{\A}^{\str \in \Pi_S, \Delta} \models \neg \varphi \mid \bar{x}_0 \in S, \Delta) < p.$$
    and $\Pi_S^k$ retains the maximal set of policies.
    \label{thm:maximal}
\end{theorem}
\begin{proof} 
    It is well known that for a history dependent unsafe policy to exist, there must be at least one stationary unsafe policy \cite{sutton2018reinforcement}.
    Therefore, we must prove that Algorithm \ref{alg:action_removal_2} removes all stationary unsafe policies for any state $\bar{x} \in \bar{X}$ whose final value $V^k(\bar{x}) < p$.
    Note that the values in our value iteration formulation are the probability of satisfying $\neg \varphi_S$ under the policy and adversary that maximizes the probability of reaching $F$, these values monotonically increase as iterations increase (excluding resets), and Algorithm \ref{alg:action_removal_2} rejects any action that results in a value $\geq p$ prior to resetting under the new constrained action space. This ensures that no stationary unsafe policy remains in the product IMDP unless the value exceeds the threshold under all actions. Hence, all remaining actions after the value iteration converges represent the set of safe policies. Since the value monotonically increases and actions are removed as soon as the value exceeds a threshold, no potentially safe actions are removed and therefore the shield retains the maximal set of safe policies upon termination. 
\end{proof}

Theorem \ref{thm:maximal} identifies that the upper bound probability of satisfying $\neg \varphi$ over all possible traces on the IMDP is less than $p$ if $V^k(\bar{x}) < p$. The following theorem extends these guarantees to the true system defined by MDP $M$.
\begin{theorem}
    [Correctness] Let $\str$ be a policy for IMDP $\I_\A$ and MDP $M$ s.t. $\str(x) = \str(q) \> \forall x \in q$. Then for $q \in Q$ if $\max_{\Delta} \mathbb{P}(\I_{\A}^{\str \in \Pi_S, \Delta} \models \neg \varphi_S \> \mid \> \bar{x}_0 = (q, z_0), \Delta) < p$
    it holds that $\mathbb{P}(M^{\str \in \Pi_S} \models \neg \varphi_S \> \mid \> x_0 \in q) < p$.
\end{theorem}

\begin{proof}
    $V^k(\bar{x})$ is the maximum probability of satisfying $\neg \varphi$ on the product IMDP $\I_{\A}$ under any policy in $\Pi_S$. The initial DFA state is $z_0$, hence $\bar{x}_0 = (q, z_0)$ denotes the initial state of a trajectory on the IMDP that begins in state $q$. As the transition probability intervals of the IMDP incorporate the uncertainty due to process noise, the modelling error of $f$, and the error induced by discretizing a continuous state space the value function upper bounds the probability that $\omega_\px \models \neg \varphi$ given $\omega_\px(0) \in q$, i.e., $V^k(\bar{x}_0) \geq \mathbb{P}(M^{\str \in \Pi_S} \models \neg \varphi_S \> \mid \> x_0 \in q)$ as follows from \cite[Theorem 2]{jackson2021formal}.
\end{proof}

We note that the resolution of the discretization impacts the resulting set of safe policies. As the volume of each partition goes to zero, the IMDP abstraction more closely models the underlying MDP (the upper/lower bound transition probabilities approach the true probability) and the set of safe policies becomes less conservative. However, increasing the resolution of the discretization results in a an exponential growth in the number of states in the IMDP, which in turn vastly increases the computational cost of Algorithm \ref{alg:action_removal_2} as shown by Theorem \ref{thm:run_time}.

\section{Examples}
    \label{examples}
    In this section we demonstrate the effectiveness of our data-driven shields. 
We consider two nonlinear systems, one with 2D and the other with 6D state spaces, under complex safety specifications. 
In each case, we set safety probability threshold $p = 0.05$ in Algorithm \ref{alg:action_removal_2} and label $\reals^n / X$ as $b$. We report the times taken to generate the IMDP abstraction and the shield in Table \ref{tab:timing}. Experiments were run on an Intel Core i7-12700K CPU at 3.60GHz with 32 GB of RAM limited to 8 threads with the shield algorithm for the 6D model run on an AMD Ryzen 7 with 128 GB of RAM. 
\begin{figure*}[t]
    \centering
    \begin{subfigure}[c]{0.3\textwidth}
        \includegraphics[width=\columnwidth]{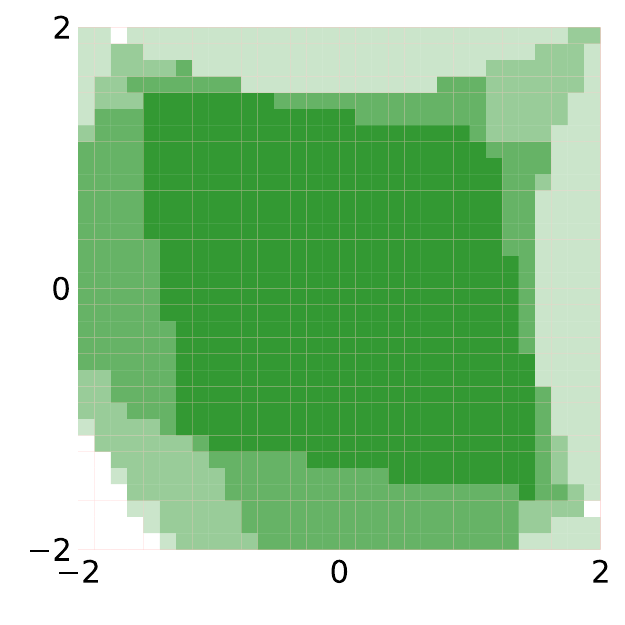}
        \caption{No Obstacles}
        \label{fig:1st env}
    \end{subfigure}
    \begin{subfigure}[c]{0.3\textwidth}
        \includegraphics[width=\columnwidth]{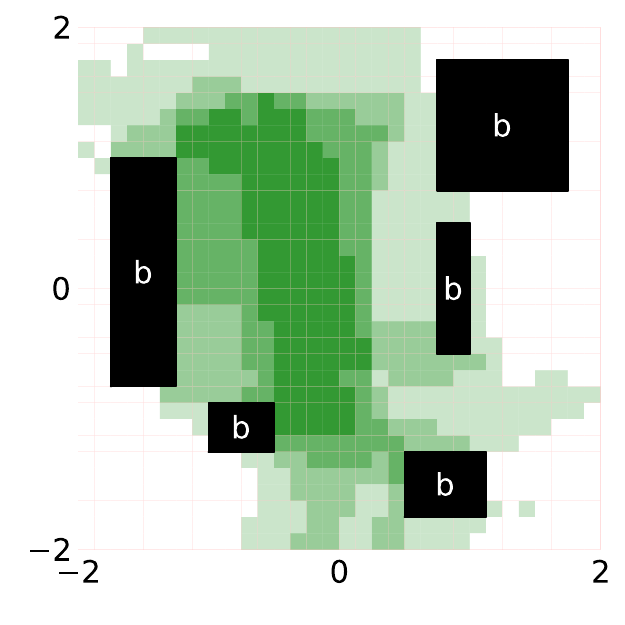}
        \caption{With Obstacles}
        \label{fig:2nd env}
    \end{subfigure}
    \begin{subfigure}[c]{0.3\textwidth}
        \includegraphics[width=\columnwidth]{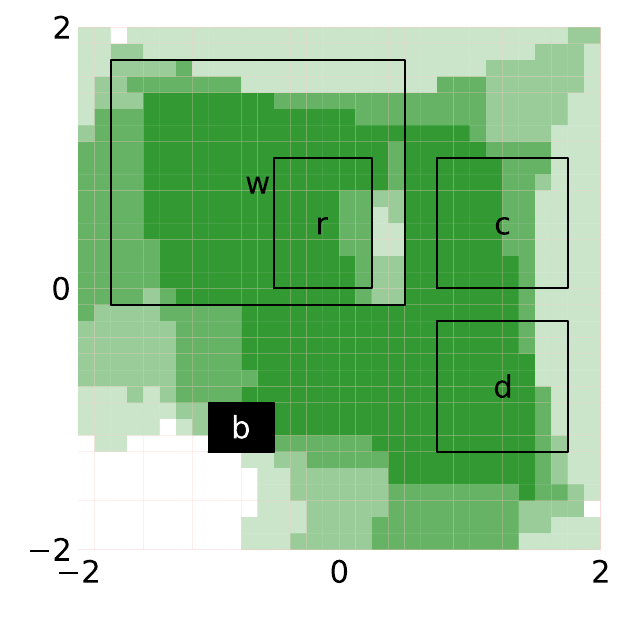}
        \caption{Complex Environment}
        \label{fig:complex_env}
    \end{subfigure}
    \begin{subfigure}[c]{0.05\textwidth}
        \includegraphics[width=\columnwidth]{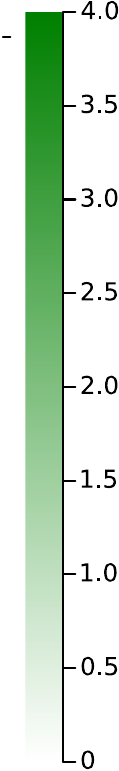}
    \end{subfigure}
    \caption{A visualization of the number of actions available at each state that can guarantee safety from the initial DFA state after using Algorithm \ref{alg:action_removal_2} on an IMDP generated from data. Dark green: all actions can be guaranteed safe to White: no action can be guaranteed safe for the given specification as shown in the legend on the right.
    }
    \label{fig:2d_visual}
\end{figure*}

\begin{table}[b]
    \centering
    \caption{Timing for IMDP and shield generation for each experiment.}
    \begin{tabular}{l  r  r}
         \toprule
         Model & Abstraction Time & Shield Time  \\
         \midrule
         2D: No Obstacles & 966.6 (sec) & 0.3 (sec) \\
         2D: Obstacles & 966.5 (sec) & 0.3 (sec) \\
         2D: Complex & 970.7 (sec) & 0.8 (sec) \\
         6D: Basilisk & 17 (hours) & 273.1 (sec)\\
         \bottomrule
    \end{tabular}
    \label{tab:timing}
\end{table}

\subsection{2D Nonlinear System}
We first consider the 2D nonlinear switched system with domain $X = [-2, 2]^2$ from \cite{reed2023promises, Adams:CSL:2022}, with additive process noise drawn from $\pv\sim\text{Uniform}(-0.01, 0.01)$. We gather 1000 data points for each mode and construct a DKL model, using all the data to train the NN and 100 data points for posterior predictions. 

To demonstrate the flexibility of our algorithm, we consider several environments and safety specifications. Figure~\ref{fig:2d_visual} shows the environments and the resulting available action space after generating a shield.
Our first scenario consists of an empty environment and a specification of $\varphi_{1} = \square(\neg b)$ (always remain in $X$) as shown in Fig.~\ref{fig:1st env}. Our second scenario is shown in Fig.~\ref{fig:2nd env}, which introduces obstacles that must be avoided and each obstacle is given the label $b$, hence this scenario retains the same safety specification $\varphi_{1}$ (``always remain in $X$ and do not collide with an obstacle"). 

Our final scenario uses the more complex environment shown in Figure~\ref{fig:complex_env} where an agent is trained to repeatedly cycle between a reward zone $r$ and a charging station $c$ without any other information of the space. The reward zone is contained inside of a wet area $w$ and a drying area is available at $d$. Our safety specification is ``if the agent is wet, do not go to the charging station unless the drying area is visited first or enough time has passed to dry off naturally (3 time steps). Also, do not collide with obstacles", which in safe LTL is
\begin{align}
    \varphi_{2} = \square \Big(w \rightarrow \big( (\neg c \, \U^{\leq 3} d) \vee (\square^{\leq 3} \neg c) \big) \Big) \wedge \square(\neg b).
\end{align}

To validate the shields, we generated 1-million random initial conditions and simulate for 1000 time steps under each shielded policy. In each case, the empirical probability of satisfying the safety specification is 100\% when shielded. In the more complex environment the system capable of visiting $r$ and $c$ infinitely often when shielding, showing that the shield is minimally invasive (contains the maximal set of safe policies) and does not impede tasking. We note that the majority of the time taken in generating the abstraction reported in Table \ref{tab:timing} is spent on calculating Post in \eqref{post_eq} and only 1-4 seconds is spent on generating the Product IMDP. Similarly, the time taken to run Algorithm \ref{alg:action_removal_2} is under 1 second for all cases.

\subsection{Space Applications}

To show scalability of our framework and its effectiveness on realistic systems, we consider a spacecraft safety application similar to \cite{Reed2024Shield}. In this scenario, an Earth observing satellite can switch between several modes of operation and a shield is designed to ensure the satellite remains functional throughout deployment. 
Spacecraft dynamics are highly nonlinear and high dimensional, with thousands of states needed to represent the interactions between subsystems. The dynamics act over continuous time and space with stochastic disturbances, which results in control problems with high complexity. 

To generate data and simulate the system, we make use of Basilisk \cite{kenneally2020basilisk}, a high-fidelity, flight-approved astrodynamics simulator which has recently been incorporated with DRL tools to identify near optimal policies for Earth observing tasks. While DRL enables a computationally tractable approach to control, the policies often lack any guarantees of safety. The goal of this experiment is to provide rigorous safety guarantees through the development of a shield. 
The action space in our scenario is $U = \{a_0, a_1, a_2, \ldots, a_{34}\}$ where $a_0$ is sun-pointing, $a_1$ is down-linking, $a_2$ is momentum dumping, and $a_3$-$a_{34}$ denote imaging modes for different upcoming targets. We generate 125000 data points under each mode using the Basilisk \cite{kenneally2020basilisk} astrodynamics simulator with a time discretization of one minute and generate a DKL model of the safety states with a 6-dimensional input space of power, normalized reaction wheel speeds, individual attitude rates ($x, y, z$), and pointing error to the sun. Stochasticity is induced by an external torque with a maximum magnitude of $10^{-4}$ Nm.

We use the safety specification $\varphi_3 =\bigcirc(\neg b) \wedge \bigcirc \bigcirc(\neg b)$
(``remain safe for the next two time steps"). This specification results in guarantees similar to the two-step shield proposed in~\cite{Reed2024Shield}, which has been shown be effective. 

\begin{figure}[t]
    \centering
    \includegraphics[width=\columnwidth]{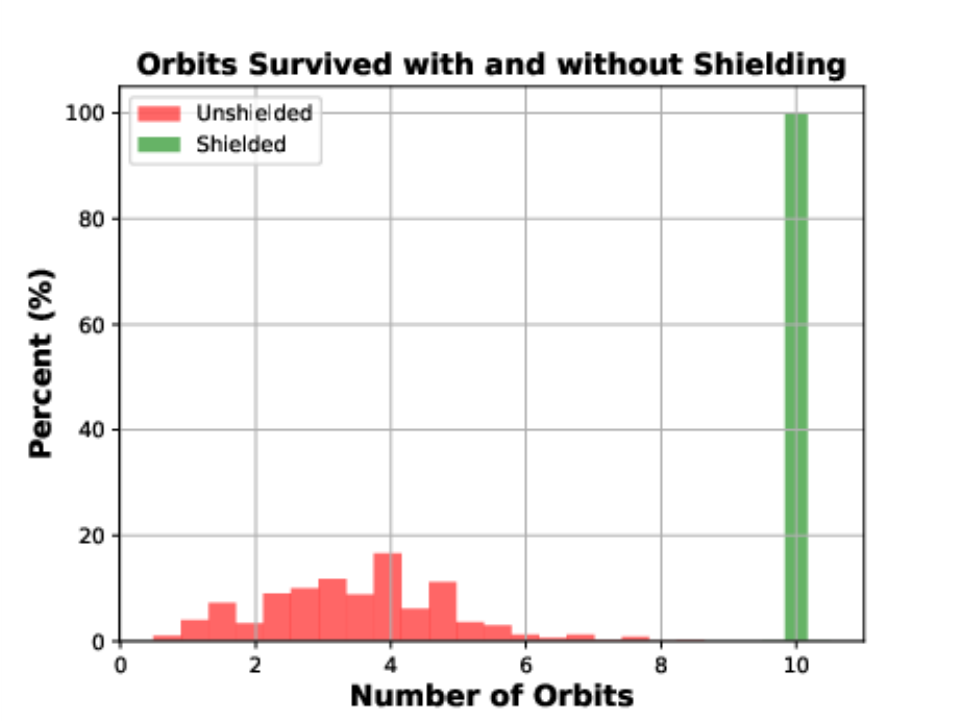}
    \caption{A histogram of Basilisk simulation end times for 1000 LEO trajectories with a shield generated by Algorithm~\ref{alg:action_removal_2} (green) and without a shield (red).}
    \label{fig:basilisk_results}
\end{figure}

We generate a fine discretization of the state space for our IMDP model, consisting of 51450 states and label the unsafe states (power $< 20$\% and normalized reaction wheel speed $> 1.2$) as $b$. The shield generated with $\varphi_3$ identifies that 59.6\% of the IMDP states are guaranteed satisfy $\varphi_3$ and 20.4\% of states cannot be verified
(i.e., there is no action that can guarantee satisfaction of $\varphi_3$). We again note that the majority of the abstraction time is spent on calculating Post, and only 1 hour of the 17 reported in Table \ref{tab:timing} is used to generate the Product IMDP. Despite the exponential time complexity of our algorithm, we identify a shield with the Product IMDP in under 5 minutes.

We deploy 1000 satellites on random LEO trajectories and simulate the evolution with Basilisk \cite{kenneally2020basilisk} under a random policy for 10 orbits, with a maximum action length of one minute. Resulting trajectory lengths are shown in Figure~\ref{fig:basilisk_results}. When deployed without the shield all trajectories result in satellite failure significantly before the 10 orbit limit, typically due to power loss. However when deployed with the shield we find that all trajectories remain safe, as noted in Figure \ref{fig:basilisk_results}.

Finally, we note that, under the much more strict safety specification of $\square(\neg b)$ (``remain safe for all time"), our algorithm identifies no safe actions.
This is likely due to uncertainty induced by the state choice (e.g. normalized wheel speeds do not indicate a direction of the vector) and could be mitigated with a higher-dimensional representation.  This gives rise to an interesting question of \emph{how to learn high-dimensional systems accurately using low dimensional models?} for future studies.

\section{Conclusion}
    \label{conclusion}
    In this work we bridge the gap in shielding literature, enabling shield generation for unknown, continuous state systems such that the system is guaranteed to satisfy a safe LTL specification. We prove that we can identify the maximal set of safe policies in finite time by generating an abstraction of the unknown system through data and that the guarantees hold for the unknown system with supporting empirical evidence. While the shields demonstrate safe deployment for high-dimensional systems such as an Earth observing satellite, more evidence is needed on the interaction of such shields and DRL policies as well as expanding guarantees to specifications of ``remain safe for all time". Similarly, the curse of dimensionality still restricts the usage of our method as abstraction generation becomes a computational burden. We hope to address these questions in future work. 

\bibliographystyle{IEEEtran}
\bibliography{references}

\begin{thebibliography}{10}
\providecommand{\url}[1]{#1}
\csname url@rmstyle\endcsname
\providecommand{\newblock}{\relax}
\providecommand{\bibinfo}[2]{#2}
\providecommand\BIBentrySTDinterwordspacing{\spaceskip=0pt\relax}
\providecommand\BIBentryALTinterwordstretchfactor{4}
\providecommand\BIBentryALTinterwordspacing{\spaceskip=\fontdimen2\font plus
\BIBentryALTinterwordstretchfactor\fontdimen3\font minus \fontdimen4\font\relax}
\providecommand\BIBforeignlanguage[2]{{%
\expandafter\ifx\csname l@#1\endcsname\relax
\typeout{** WARNING: IEEEtran.bst: No hyphenation pattern has been}%
\typeout{** loaded for the language `#1'. Using the pattern for}%
\typeout{** the default language instead.}%
\else
\language=\csname l@#1\endcsname
\fi
#2}}

\bibitem{bloem2015shield}
R.~Bloem, B.~K{\"o}nighofer, R.~K{\"o}nighofer, and C.~Wang, ``Shield synthesis: Runtime enforcement for reactive systems,'' in \emph{TACAS}.\hskip 1em plus 0.5em minus 0.4em\relax Springer, 2015.

\bibitem{alshiekh2018safe}
M.~Alshiekh, \emph{et~al.}, ``Safe reinforcement learning via shielding,'' in \emph{AAAI conference on artificial intelligence}, vol.~32, no.~1, 2018.

\bibitem{harris2020spacecraft}
A.~T. Harris and H.~Schaub, ``Spacecraft command and control with safety guarantees using shielded deep reinforcement learning,'' in \emph{AIAA Scitech 2020 Forum}, 2020, p. 0386.

\bibitem{nazmy2022shielded}
I.~Nazmy, A.~Harris, M.~Lahijanian, and H.~Schaub, ``Shielded deep reinforcement learning for multi-sensor spacecraft imaging,'' in \emph{2022 American Control Conference (ACC)}.\hskip 1em plus 0.5em minus 0.4em\relax IEEE, 2022, pp. 1808--1813.

\bibitem{Reed2024Shield}
R.~Reed, H.~Schaub, and M.~Lahijanian, ``Shielded deep reinforcement learning for complex spacecraft tasking,'' in \emph{ACC}, 2024.

\bibitem{yang2023safe}
W.-C. Yang, G.~Marra, G.~Rens, and L.~De~Raedt, ``Safe reinforcement learning via probabilistic logic shields,'' in \emph{IJCAI}, 2023.

\bibitem{jansen2020safe}
N.~Jansen, \emph{et~al.}, ``Safe reinforcement learning using probabilistic shields,'' in \emph{CONCUR}.\hskip 1em plus 0.5em minus 0.4em\relax Schloss-Dagstuhl-Leibniz Zentrum f{\"u}r Informatik, 2020.

\bibitem{reed2023promises}
R.~Reed, L.~Laurenti, and M.~Lahijanian, ``Promises of deep kernel learning for control synthesis,'' \emph{IEEE Control Systems Letters}, 2023.

\bibitem{Ober2021}
S.~W. Ober, C.~E. Rasmussen, and M.~van~der Wilk, ``The promises and pitfalls of deep kernel learning,'' in \emph{Uncertainty in Artificial Intelligence}.\hskip 1em plus 0.5em minus 0.4em\relax PMLR, 2021.

\bibitem{kupferman2001model}
O.~Kupferman and M.~Y. Vardi, ``Model checking of safety properties,'' \emph{Formal methods in system design}, vol.~19, pp. 291--314, 2001.

\bibitem{mazouz2024data}
R.~Mazouz, \emph{et~al.}, ``Data-driven permissible safe control with barrier certificates,'' \emph{arXiv preprint arXiv:2405.00136}, 2024.

\bibitem{wajid2022formal}
R.~Wajid, A.~U. Awan, and M.~Zamani, ``Formal synthesis of safety controllers for unknown stochastic control systems using gaussian process learning,'' in \emph{L4DC}.\hskip 1em plus 0.5em minus 0.4em\relax PMLR, 2022.

\bibitem{jagtap2020formal}
P.~Jagtap, S.~Soudjani, and M.~Zamani, ``Formal synthesis of stochastic systems via control barrier certificates,'' \emph{IEEE TAC}, 2020.

\bibitem{santoyo2021barrier}
C.~Santoyo, M.~Dutreix, and S.~Coogan, ``A barrier function approach to finite-time stochastic system verification and control,'' \emph{Automatica}, 2021.

\bibitem{Chowdhury}
S.~R. Chowdhury and A.~Gopalan, ``On kernelized multi-armed bandits,'' in \emph{ICML}.\hskip 1em plus 0.5em minus 0.4em\relax PMLR, 2017.

\bibitem{abbasi2013online}
Y.~Abbasi-Yadkori, ``Online learning for linearly parametrized control problems,'' 2013.

\bibitem{steinwart2001influence}
I.~Steinwart, ``On the influence of the kernel on the consistency of support vector machines,'' \emph{JMLR}, vol.~2, no. Nov, pp. 67--93, 2001.

\bibitem{lahijanian2011control}
M.~Lahijanian, S.~Andersson, and C.~Belta, ``Control of markov decision processes from pctl specifications,'' in \emph{ACC}.\hskip 1em plus 0.5em minus 0.4em\relax IEEE, 2011, pp. 311--316.

\bibitem{baier2008principles}
C.~Baier and J.-P. Katoen, \emph{Principles of model checking}.\hskip 1em plus 0.5em minus 0.4em\relax MIT press, 2008.

\bibitem{jackson2021formal}
J.~Jackson, L.~Laurenti, E.~Frew, and M.~Lahijanian, ``Formal verification of unknown dynamical systems via gaussian process regression,'' \emph{arXiv preprint arXiv:2201.00655}, 2021.

\bibitem{rasmussen:book:2006}
C.~E. Rasmussen, C.~K. Williams, \emph{et~al.}, \emph{Gaussian processes for machine learning}.\hskip 1em plus 0.5em minus 0.4em\relax Springer, 2006, vol.~1.

\bibitem{givan2000bounded}
R.~Givan, S.~Leach, and T.~Dean, ``Bounded-parameter markov decision processes,'' \emph{Artificial Intelligence}, vol. 122, no.~1, 2000.

\bibitem{Lahijanian:TAC:2015}
M.~Lahijanian, S.~B. Andersson, and C.~Belta, ``Formal verification and synthesis for discrete-time stochastic systems,'' \emph{IEEE TAC}, vol.~60, 2015.

\bibitem{reed2024error}
R.~Reed, L.~Laurenti, and M.~Lahijanian, ``Error bounds for gaussian process regression under bounded support noise with applications to safety certification,'' \emph{arXiv preprint arXiv:2408.09033}, 2024.

\bibitem{skovbekk2023formal}
J.~Skovbekk, L.~Laurenti, E.~Frew, and M.~Lahijanian, ``Formal abstraction of general stochastic systems via noise partitioning,'' \emph{IEEE Control Systems Letters}, 2023.

\bibitem{sutton2018reinforcement}
R.~S. Sutton, ``Reinforcement learning: An introduction,'' \emph{A Bradford Book}, 2018.

\bibitem{Adams:CSL:2022}
S.~A. Adams, M.~Lahijanian, and L.~Laurenti, ``Formal control synthesis for stochastic neural network dynamic models,'' \emph{IEEE Control Systems Letters}, 2022.

\bibitem{kenneally2020basilisk}
P.~W. Kenneally, S.~Piggott, and H.~Schaub, ``Basilisk: A flexible, scalable and modular astrodynamics simulation framework,'' \emph{Journal of aerospace information systems}, vol.~17, no.~9, pp. 496--507, 2020.

\end{thebibliography}

\end{document}